\documentclass[conference,a4paper]{IEEEtran}

\usepackage{amsmath, amsthm, amssymb}
\usepackage{amsfonts}
\usepackage{amsthm}
\usepackage{tikz}
\usetikzlibrary{shapes,arrows}
\usepackage{verbatim}
\usepackage[]{subfigure}
\theoremstyle{remark}

\usepackage{graphicx,relsize,verbatim,enumerate,xcolor}
\usepackage{pgfplots}

\usepackage{hyperref}

\hypersetup{
    colorlinks=true,
    linkcolor=black,
    citecolor=black,
    filecolor=black,
    urlcolor=black,
    bookmarks=false,
    pdfstartview=FitH,
}

\usetikzlibrary{arrows,positioning,fit,backgrounds,shapes}
\usetikzlibrary{calc}

\newtheorem{lemma}{Lemma}
\newtheorem{theorem}{Theorem}

\newtheorem{corollary}{Corollary}

\tikzstyle{arw}=[->,>=latex]
\tikzstyle{node}=[rectangle,draw,outer sep=0pt,minimum width=1.7cm, minimum height=8mm]

\newcommand{\rank}{{\mbox{rank}}}

\newcommand{\E}{{\mathbb{E}}}
\newcommand{\br}[1]{{\left(#1\right)}}
\newcommand{\setbr}[1]{{\left\{#1\right\}}}

\newcommand{\sq}[1]{{\left[#1\right]}}
\newcommand{\R}{{\mathbb{R}}}

\newcommand{\No}{{\mathcal{N}}}

\newcommand{\diag}{{\mbox{diag}}}
\newcommand{\inv}{{^{-1}}}

\newcommand{\la}{{\lambda}}

\newcommand{\Si}{{\Sigma}}

\newcommand{\X}{{\bar{X}}}
\newcommand{\Y}{{\bar{Y}}}

\newcommand{\Xt}{{\tilde{X}}}
\newcommand{\Yt}{{\tilde{Y}}}

\newcommand{\As}{\textsf{A}~}
\newcommand{\Bs}{\textsf{B}~}

\newcommand{\Es}{\textsf{E}~}

\newcommand{\0}{{\mathbf{0}}}

\newcommand{\es}{{\emptyset}}

\begin{document}

\sloppy

\title{Gaussian Secure Source Coding and Wyner's Common Information} 

\author{
  \IEEEauthorblockN{Sanket Satpathy and Paul Cuff}
  \IEEEauthorblockA{Dept. of Electrical Engineering\\
    Princeton University\\
    Princeton, USA\\
    Email: \{satpathy,cuff\}@princeton.edu}
}

\maketitle

\begin{abstract}
  We study secure source-coding with causal disclosure, under the Gaussian distribution.  The optimality of Gaussian auxiliary random variables is shown in various scenarios. We explicitly characterize the tradeoff between the rates of communication and secret key.  This tradeoff is the result of a mutual information optimization under Markov constraints.  As a corollary, we deduce a general formula for Wyner's Common Information in the Gaussian setting.
\end{abstract}


\section{Introduction}

There is a growing body of work in secure source coding \cite{vp,gund,tand,kc1,kc2,Cuff3,rdss}. 
Most of the problem formulations consider distortion at the legitimate receiver and equivocation at the eavesdropper. An alternative approach was proposed by Yamamoto \cite{yam1,yam2} which replaced the eavesdropper's equivocation with the distortion incurred by the eavesdropper's best estimate of the information source.  The motivation behind this formulation is a purely operational approach to the problem of secrecy.  The choice of distortion function may depend on the context in which secrecy is desired.

Recently, the problem posed by Yamamoto was solved \cite{Cuff5} and considerably generalized \cite{Cuff3,rdss,Cuff4}.  The salient feature of the new approach is the causal disclosure of information to the eavesdropper.  There are compelling arguments that this disclosure is necessary for a robust notion of secure communication \cite{rdss}.  This formulation of secrecy is natural when understood in a game-theoretic context. A repeated game is being played by the adversary versus the communication system. Distortion is replaced by a payoff function, while the information sequences equate to actions of the players.

Remarkably, when the payoff is chosen to be the log-loss function \cite{court}, the above framework recovers results for (normalized) equivocation-based secrecy \cite{eq}. Under this choice of payoff, the adversary expresses her belief about the distribution of the current information symbol, given her knowledge of past symbols.  Thus, the secure source coding framework of \cite{rdss} generalizes traditional approaches to secrecy.  We note that with this generality, new challenges arise in certain contexts involving uncoded side information at the receiver \cite{pt1,pt2,me}.

\tikzstyle{block} = [draw, fill=blue!20, rectangle, 
    minimum height=2em, minimum width=4em]
\tikzstyle{sum} = [draw, fill=blue!20, circle, node distance=1cm]
\tikzstyle{input} = [coordinate]
\tikzstyle{output} = [coordinate]
\tikzstyle{pinstyle} = [pin edge={to-,thin,black}]

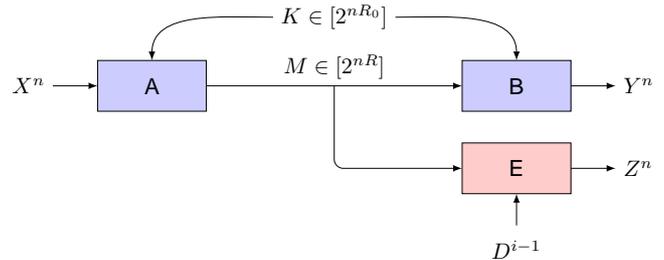
\begin{figure}[ht]
\begin{center}
 \resizebox{3.4in}{!}{\begin{tikzpicture}
 [node distance=1cm,minimum width=1cm,minimum height =.75 cm]
  \node[rectangle,minimum width=5mm] (source) {$X^n$};
  \node[node, fill=blue!20] (alice) [right =7mm of source] {\As};
  \node[node, fill=blue!20] (bob) [right =4cm of alice] {\Bs};
  \node[coordinate] (dummy) at ($(alice.east)!0.5!(bob.west)$) {};
  \node[rectangle,minimum width=5mm] (xhat) [right =7mm of bob] {$Y^n$};
  \node[rectangle,minimum width=7mm] (key) [above =7mm of dummy] {$K\in[2^{nR_0}]$};
  \node[node, fill=red!20] (eve) [below =5mm of bob] {\Es};
  \node[rectangle,minimum width=5mm] (zn) [right =7mm of eve] {$Z^n$};
  \node[rectangle] (side) [below=5mm of eve] {$D^{i-1}$};
  
  

  \draw [arw] (source) to (alice);
  \draw [arw] (alice) to node[minimum height=6mm,inner sep=0pt,midway,above]{$M\in[2^{nR}]$} (bob);
  \draw [arw] (bob) to (xhat);
  \draw [arw] (key) to [out=180,in=90] (alice);
  \draw [arw] (key) to [out=0,in=90] (bob);
  \draw [arw,rounded corners] (dummy) |- (eve);
  \draw [arw] (eve) to (zn);
  \draw [arw] (side) to (eve);
 \end{tikzpicture}}
 \caption{\small The causal disclosure framework for secure source coding \cite{rdss}. Disclosures $D=(D_{x},D_{y})$ are allowed, with arbitrary orthogonal disclosure channels $P_{D_{x},D_{y}|XY}=P_{D_{x}|X}P_{D_{y}|Y}$.}
 \label{pstsh}
 \end{center}
 \end{figure}
 
 Now, we recall the main result of \cite{rdss} derived under this framework:  The optimal tradeoff between communication rate $R$, secret key rate $R_0$ and average payoff $\Pi$ is given by the union of regions
 \begin{align}
R&\ge I(X;U,V),\label{sc}\\
R_0&\ge I(D_x,D_y;V|U),\label{sc1}\\
\Pi&\le \min_{z(\cdot)}\E\sq{\pi(X,Y,z(U))},\label{sc2}
\end{align}
where the union is taken over distributions that enforce the Markov chain $D_x-X-(U,V)-Y-D_y$.  Though the presentation in \cite{rdss} restricts itself to the case when all random variables have finite alphabets, most of the results (lossless communication is an exception) easily generalize to continuous random variables.

The scheme used to obtain the above region admits a simple interpretation.  The information source $X^n$ is split into two parts:  a secure part $V^n$ and a non-secure part $U^n$.  The eavesdropper is given full knowledge of $U^n$, while the secret key is focused on keeping $V^n$ perfectly secure.  This can be implemented using a superposition code \cite{rdss}.

Unlike $V$, the variable $U$ plays a specific and concrete role in the secure communication system as the information that is leaked to the eavesdropper, discussed in \cite[Section VI-F]{rdss}.  That is, the significance of the distribution of $U$ is more than simply that of an optimization parameter for the region (1)-(3).

This paper asks the following question:  Given that $P_{X,Y,U}$ is Gaussian (this fixes the bound on $\Pi$), can the communication-key tradeoff be realized with Gaussian $P_{V|X,Y,U}$?  Our primary motivation for this investigation is a potential application to the problem of secure rate-limited control.  In the context of control,  $Y^n$ may be a Gaussian control signal that is correlated with the state process $X^n$, and $U^n$ is a Gaussian degradation of the control that is leaked to the eavesdropper.

Classical control theory \cite{bert} provides exact characterizations of control performance for Gauss-Markov processes.  If the relevant rates are optimized by Gaussian distributions, then we can replace rate-limited feedback links with idealized Gaussian channels and use these characterizations to derive tight bounds on performance. This observation has already been used by Tatikonda-Mitter-Sahai \cite{tat,tat1} to characterize optimal performance in rate-limited control with quadratic costs.  

It is worth pointing out that if the optimization is carried out jointly over $(U,V)$ satisfying the Markov constraint, then Gaussian $P_{U,V|X,Y}$ does not suffice to achieve the entire rate-payoff region even when $\pi(\cdot,\cdot,\cdot)$ is a quadratic function \cite{eva}.  However, the Gaussianity of $U$ is motivated by operational considerations derived from the coding scheme described above.  Since $U^n$ represents information that is revealed to the eavesdropper, this is conveniently modeled in many applications by a linear/additive channel from $X^n$ or $Y^n$ to $U^n$.  Such degradations can be often be realized by physical processes (e.g. optical, electrical).  Further theoretical justification is provided by the worst-additive-noise-lemma \cite{gauss1,digg} when the payoff is pointwise mutual information $\pi(y,z)=-\log \frac{p(y)p(z)}{p(y,z)}$.

Besides the potential application to secure control, the problem we consider is interesting in its own right as a mutual information optimization under unusual Markov constraints.  There has been much effort in the information theory community focused on proving optimality of Gaussian random variables for various applications \cite{gauss1,digg,gauss2}.  We remark that recent techniques \cite{gauss2} designed to prove the optimality of Gaussian auxiliaries seem to be best suited to cases where the optimization is over random variables at the extremes of Markov chains \cite{courtade}.  It is unclear if the method can be adapted to our setting, where the optimization is over an auxiliary in the middle of a Markov chain.

Our approach will be a strengthening of the estimation-theoretic technique used to compute the common information of a bivariate Gaussian distribution in \cite{biao1} (the result first appeared in \cite{biao2}, but the proof had a gap that was later corrected).  As a corollary, we deduce a general formula for Wyner's common information in the Gaussian setting.  This quantity has proved to be fundamental in various source coding problems \cite{biao1,Wyner,DCS,reverse2,lossyci}, although most of these results consider sources with a finite alphabet.

\section{Notation}


We represent both random variables and probability distribution functions with capital letters, but only letters $P$ and $Q$ are used for the latter. The set of real numbers is denoted by $\R$, while $\R_+$ denotes non-negative reals.  We denote the conditional distribution of the random variable $Y$ given the random variable $X$ by $P_{Y|X}(y|x)$. This is the usual notation,  although sometimes we do abbreviate it as $P_{Y|X}$. Markov chains are denoted by $X-Y-Z$ implying the factorization $P_{X,Y,Z}=P_{X,Y}P_{Z|Y}$ while $X \perp Y$ indicates that the random variables $X$ and $Y$ are independent.
Sequences of random variables $X_1,\ldots,X_n$ are denoted by $X^n$.

Let $\diag(\setbr{a_i}_{i=1}^r)\in\R^{r\times r}$ denote a diagonal matrix $A$ with diagonal entries $A_{ii}=a_i$. The transpose of a matrix $A$ is denoted by $A^T$ and $A^{-T}=(A\inv)^T$.  If $X\in\R^d$ is a $d$-dimensional (column) vector, then $X_{:i}$ denotes the vector formed by the first $i$ components of $X$. 

We denote the covariance matrix of zero-mean random vectors $X\in\R^d$ by $\Si_X:=\E[XX^T]\in\R^{d\times d}$. When $d=1$, we set $\sigma_X^2:=\Si_X$.  For zero-mean random vectors $X,Y\in\R^d$, the cross-covariance matrix is denoted by $\Si_{XY}:=\E[XY^T]\in\R^{d\times d}$. Note that $\Si_{XY}=\Si_{YX}^T$.

Recall that $\Si_X$ is real, symmetric and positive semi-definite.  Its eigen-decomposition is given by $\Si_X=B_X\Lambda B_X^T$.  Let $r_X:=\rank(\Si_X)$.
We set $\Si_X^{1/2}=B_X\Lambda^{1/2} B_X^T$ and $\Si_X^{-1/2}=B_X\Lambda^{-1/2} B_X^T$, with 
\begin{equation}
\Lambda^{-1/2}=\br{\begin{array}{c c}
\Lambda_+^{-1/2}&\0\\
\0&\0
\end{array}},
\end{equation} 
where $\Lambda_+\in\R^{r_X\times r_X}$ is the submatrix of $\Lambda$ with strictly positive diagonal entries.  Note that we have
\begin{equation}
\Si_X^{-1/2}\Si_X\Si_X^{-1/2}=\br{\begin{array}{c c}
I_r&\0\\
\0&\0
\end{array}},
\end{equation}
where $I_r\in\R^{r\times r}$ is the identity matrix.

\section{Main Result}
In the following, we assume that $X,Y,U\in \R^d$ are jointly Gaussian random vectors.  There is no loss of generality in assuming the same length for all vectors since we can zero-pad shorter vectors.  For simplicity, we restrict to disclosures $D_x\in\setbr{\es,X}$ and $D_y\in\setbr{\es,Y}$.

\begin{theorem}
For jointly Gaussian $X,Y,U\in \R^d$, the region \eqref{sc}-\eqref{sc2} is optimized by Gaussian $P_{V|XUY}$, such that $X-(U,V)-Y$ holds.  In particular, we have the following communication-key tradeoffs:
\begin{itemize}
\item Arbitrary $P_{D_x|X}$ and $D_y=\es$:
\begin{align}
R&\ge I(X;U,Y),\\
R_0&\ge I(D_x;Y|U),
\end{align}

\item $D_x=\es,D_y=Y$:
\begin{align}
R&\ge \frac{1}{2}\sum_{i=1}^d\E_U\sq{\log\frac{1}{(1-\rho_{XU,i}^2)(1-a_{\la,i}^U)}},\\
R_0&\ge \frac{1}{2}\sum_{i=1}^d\E_U\sq{\log\frac{1}{\br{1-\frac{\rho_{XY|U,i}^2}{a_{\la,i}^U}}}},
\end{align}
with 
\begin{equation}
a_{\la,i}^U=\frac{(\la-1)\rho_{XY|U,i}^2+\sqrt{4\la\rho_i^2+(1-\la)^2\rho_{XY|U,i}^4}}{2\la},\label{a1}
\end{equation}

\item $D_x=X,D_y=Y$:
\begin{align}
R&\ge \frac{1}{2}\sum_{i=1}^d\E_U\sq{\log\frac{1}{(1-\rho_{XU,i}^2)(1-a_{\la,i}^U)}},\label{res1}\\
R_0&\ge \frac{1}{2}\sum_{i=1}^d\E_U\sq{\log\frac{1-\rho_{XY|U,i}^2}{(1-a_{\la,i}^U)\br{1-\frac{\rho_{XY|U,i}^2}{a_{\la,i}^U}}}},\label{res2}
\end{align}
with 
\begin{equation}
a_{\la,i}^U=\frac{\la\rho_{XY|U,i}^2+\rho_{XY|U,i}\sqrt{\la^2\rho_{XY|U,i}^2+4(\la+1)}}{2(\la+1)},\label{a2}
\end{equation}
\end{itemize}
where $\setbr{\rho_{XU,i}}_{i=1}^d$ are singular values of $\Si_{X}^{-1/2}\Si_{XU}\Si_{U}^{-1/2}$, $\setbr{\rho_{XY|U,i}}_{i=1}^d$ are singular values of $\Si_{X|U}^{-1/2}\Si_{XY|U}\Si_{Y|U}^{-1/2}$ and the tradeoffs are parametrized by $\la\in[0,\infty)$.
\label{thm}\end{theorem}

\subsection{Interpretation}

The parameter $\la$ specifies the $(R,R_0)$ point that is tangent to a supporting line with slope $-\la\inv$.  Equations \eqref{a1} and \eqref{a2} precisely capture the path traced by the optimizing channels $P_{X|U,V}$ and $P_{Y|U,V}$ as the line is varied.
Note that the first case ($D_y=\es$) immediately follows from the data-processing inequality.

The above results are expressed in terms of the singular values of the \emph{correlation matrix} $\rho_{XY}:=\Si_X^{-1/2}\Si_{XY}\Si_Y^{-1/2}$. Recall that the linear MMSE estimator is given in terms of this matrix (for zero-mean random variables) as
\begin{equation}
\E[X|Y]=\Si_X^{1/2}\rho_{XY}\Si_Y^{-1/2}Y.
\end{equation}
In the scalar case, this simplifies to
\begin{equation}
\E[X|Y]=\rho_{XY}\frac{\sigma_X}{\sigma_Y}Y,
\end{equation}
where $\rho_{XY}$ is the correlation coefficient.

We also generalize a result of \cite{biao1}, which considered the case of scalar Gaussian random variables.

\begin{corollary}
For jointly Gaussian $X,Y\in\R^d$, Wyner's common information is given by
\begin{equation}
C(X;Y):=\min_{U:X-U-Y}I(X,Y;U)=\frac{1}{2}\sum_{i=1}^d\log\frac{1+\rho_i}{1-\rho_i},
\end{equation}
where $\setbr{\rho_{i}}_{i=1}^d$ are singular values of $\Si_{X}^{-1/2}\Si_{XY}\Si_{Y}^{-T/2}$.
\end{corollary}

In the following sections, we only present the proof for $(D_x,D_y)=(X,Y)$.  The proof for the second case is similar and thus, omitted.

\section{Proof}

Note that the communication rate $I(X;U,V)$ is minimized by $V=Y$, while the secret key rate $I(X,Y;V|U)$ is minimized by choosing $V$ such that $I(X,Y;V|U=u)=C(X;Y|U=u)$ for every $U=u$.  In general, it is not possible to minimize both rates simultaneously.

We consider the optimal frontier of rates by considering the point at which a supporting hyperplane touches the region.  In other words, we would like to show that
\begin{equation}
\underset{P_{V|XUY}:X-(U,V)-Y}{\mbox{argmin}}\br{\la I(X;U,V)+I(X,Y;V|U)}
\end{equation}
is minimized by a Gaussian distribution for $\la\ge0$.  We shall constructively show that a minimizer exists, so the above expression is well-defined.  In the following, we shall perform the analysis conditioned on $U$, so it suffices to consider the problem
\begin{equation}
\underset{P_{V|X,Y}:X-V-Y}{\mbox{argmin}}\br{\la I(X;V)+I(X,Y;V)}.
\end{equation}
Since $P_{X,Y}$ is Gaussian, linear-Gaussian $P_{V|X,Y}$ ensures that the joint distribution is Gaussian as well.  Note that the Gaussianity of $V$ is not necessary, since mutual information is invariant under invertible transformations.

\subsection{Diagonalization}

Consider $\X=\Si_X^{-1/2}X$ and $\Y=\Si_Y^{-1/2}Y$.  This is defined to mean that only the positive eigenvalues are inverted.  The zero eigenvalues remain zero.  However, this is still a matrix inverse i.e.\ $X = \Sigma_X^{1/2} \Sigma_X^{-1/2} X$ with probability one.
We have
\begin{align}
\Si_{\X}&=\br{\begin{array}{c c}
I_{r_X}&0\\
0&0_{d-r_X}
\end{array}},\\
\Si_{\Y}&=\br{\begin{array}{c c}
I_{r_Y}&0\\
0&0_{d-r_Y}
\end{array}}.
\end{align}

Also, we have
\begin{equation}
\Si_{\X\Y}=\Si_X^{-1/2}\Si_{XY}\Si_Y^{-1/2}=\br{\begin{array}{c c}
A_{X,Y}&0\\
0&0_{d-r_X,d-r_Y}
\end{array}},
\end{equation}
where $A_{X,Y}\in\R^{r_X\times r_Y}$.  By the singular value decomposition, we have
\begin{equation}
A_{X,Y}=B_X\Lambda B_Y,
\end{equation}
where $\Lambda\in\R_+^{r_X\times r_Y}$ is diagonal and $B_X\in\R^{r_X\times r_X}, B_Y\in\R^{r_Y\times r_Y}$ are orthogonal matrices.  Then with
\begin{align}
\Xt&=\br{\begin{array}{c c}
B_X&0\\
0&0_{d-r_X}
\end{array}}\X,\\
\Yt&=\br{\begin{array}{c c}
B_Y&0\\
0&0_{d-r_Y}
\end{array}}\Y,
\end{align}
we have $\Si_\Xt=\Si_\X$, $\Si_\Yt=\Si_\Y$ and
\begin{equation}
\Si_{\Xt\Yt}=\br{\begin{array}{c c}
\Lambda&0\\
0&0_{d-r_X,d-r_Y}
\end{array}},
\end{equation}
where the non-zero diagonal entries of $\Lambda$ are $\setbr{\rho_i}_{i=1}^r$ ($\rho_i>0$), the singular values of the \emph{correlation matrix} $\Si_X^{-1/2}\Si_{XY}\Si_Y^{-1/2}$.

Thus, we have constructed invertible linear transformations $X\mapsto \Xt$ and $Y\mapsto \Yt$ such that $\Si_\Xt, \Si_\Yt$ and $\Si_{\Xt\Yt}$ are diagonal.  Since mutual information is invariant to invertible transformations, it suffices to show that
\begin{equation}
\arg\min_{P_{V|\Xt\Yt}:\Xt-V-\Yt}\br{\la I(\Xt;V)+I(\Xt,\Yt;V)}
\end{equation}
is Gaussian.

\subsection{Achievability Proof}

Let $r:=\min(r_X,r_Y)$.  Consider independent random variables $V,Z_1,Z_2\sim\No(0,I_r)$. Let $A_\la=\diag(\setbr{a_{\la,i}}_{i=1}^r)$ and $B_\la=\diag(\setbr{b_{\la,i}}_{i=1}^r)$. Consider ($0\le a_{\la,i},b_{\la,i} \le 1$ for all $i$)
\begin{align}
\Xt_{:r}&=\sqrt{A_\la}V+\sqrt{I-A_\la}Z_1\\
\Yt_{:r}&=\sqrt{B_\la}V+\sqrt{I-B_\la}Z_2,
\end{align}
with
$\sqrt{a_{\la,i} b_{\la,i}}=\rho_i$, so that $P_{\Xt\Yt}$ is realized.  It suffices to generate the remaining $(d-r)$ components of $\Xt$ and $\Yt$ independent of $V$. Note that we have $\Xt-V-\Yt$ and $(\Xt,\Yt)\sim P_{\Xt\Yt}$ under the above construction.

Under this distribution, we have (since $\Xt_i-V_i-\Yt_i~ \forall i$)
\begin{align}
&~\quad\la I(\Xt;V)+I(\Xt\Yt;V)\\
&=\la I(\Xt_{:r};V)+I(\Xt_{:r}\Yt_{:r};V)\\
&=\sum_{i=1}^r\br{\la I(\Xt_i;V_i)+I(\Xt_i\Yt_i;V_i)}\\
&=\sum_{i=1}^r\br{(\la+1) I(\Xt_i;V_i)+I(\Yt_i;V_i)-I(\Xt_i;\Yt_i)}\\
&=\sum_{i=1}^r\br{\frac{1}{2}\log\frac{1}{(1-a_{\la,i})^{\la+1}(1-b_{\la,i})}-\frac{1}{2}\log\frac{1}{(1-\rho_i^2)}}\\
&=\sum_{i=1}^r\br{\frac{1}{2}\log\frac{a_{\la,i}}{(1-a_{\la,i})^{\la+1}(a_{\la,i}-\rho_i^2)}-\frac{1}{2}\log\frac{1}{(1-\rho_i^2)}}.
\end{align}
Now, set
\begin{equation}
a_{\la,i}=\frac{\la\rho_i^2+\rho_i\sqrt{\la^2\rho_i^2+4(\la+1)}}{2(\la+1)}\label{adef1}
\end{equation}
to achieve \eqref{res1}-\eqref{res2}. It is easy to check that $\rho_i^2\le a_{\la,i}\le 1$, which respects the correlation constraint $\sqrt{a_{\la,i} b_{\la,i}}=\rho_i$.
Note that $a_{0,i}=\rho_i$, which recovers the construction of \cite{biao1}.  Also, $\lim_{\la\to\infty}a_{\la,i}=\rho_i^2$, which reflects the fact that $V_i= \Yt_i$ with probability 1 as $\la\to\infty$.

\subsection{Converse Proof}

The following lemma shall be crucial in establishing the optimality of our construction.  This is essentially the Cauchy-Schwarz inequality.  In \cite{biao1}, the AM-GM inequality and the orthogonality principle in optimal estimation were used in the converse proof.  This approach does not work here due to the asymmetry introduced by the parameter $\la$.

\begin{lemma}
For unit variance random variables $X,Y\in\R$ and any $P_{V|X,Y}$ such that $X-V-Y$,
we have
\begin{equation}
\rho^2=\E[XY]^2\le \E\sq{\E[X|V]^2}\E\sq{\E[Y|V]^2}
\end{equation}
\end{lemma}

\begin{proof}
By $X-V-Y$ and the Cauchy-Schwarz inequality, we have
\begin{align}
\E[XY]^2&=\E_V[\E[XY|V]]^2\\
&=\E_V[\E[X|V]\E[Y|V]]^2\\
&\le\E[\E[X|V]^2]\E[\E[Y|V]^2].
\end{align}
\end{proof}

Consider any $P_{V|XY}$ such that $\Xt-V-\Yt$ (so $\Xt_i-V-\Yt_i$ holds for all $i$). Note that we don't make any structural assumptions on $V$ here.  Let $D_{\Xt_i}:=\E[(\Xt_i-\E[\Xt_i|V])^2]$ and $D_{\Yt_i}:=\E[(\Yt_i-\E[\Yt_i|V])^2]$.  Using standard information-theoretic inequalities, we have
\begin{align}
&\quad~\la I(\Xt;V)+I(\Xt\Yt;V)\\
&\ge\la I(\Xt_{:r};V)+I(\Xt_{:r}\Yt_{:r};V)\\
&=\la \sum_{i=1}^r I(\Xt_i;V|\Xt^{i-1})+\sum_{i=1}^r I(\Xt_i\Yt_i;V|\Xt^{i-1}\Yt^{i-1})\\
&=\la \sum_{i=1}^r I(\Xt_i;V,\Xt^{i-1})+\sum_{i=1}^r I(\Xt_i\Yt_i;V,\Xt^{i-1}\Yt^{i-1})\\
&\ge\sum_{i=1}^r\br{\la I(\Xt_i;V)+I(\Xt_i\Yt_i;V)}\\
&=\sum_{i=1}^r\br{(\la+1) I(\Xt_i;V)+I(\Yt_i;V)-I(\Xt_i;\Yt_i)}\\
&\ge\sum_{i=1}^r\br{(\la+1) I(\Xt_i;\E[\Xt_i|V])+I(\Yt_i;\E[\Yt_i|V])-I(\Xt_i;\Yt_i)}\label{dpi}\\
&\ge\sum_{i=1}^r\br{(\la+1) R_{\Xt_i}(D_{\Xt_i})+R_{\Yt_i}(D_{\Yt_i})-I(\Xt_i;\Yt_i)}\label{grd}\\
&=\sum_{i=1}^r\br{\frac{1}{2}\log\frac{1-\rho_i^2}{D_{\Xt_i}^{\la+1}D_{\Yt_i}}},\label{conv1}
\end{align}
where \eqref{dpi} follows from the data-processing inequality and \eqref{grd} follows from the definition of the Gaussian rate-distortion function $R_X(\cdot)$ \cite[Theorem 10.3.2]{Cover}.

Since $\Xt_i-V-\Yt_i$ holds for all $i$, using 
\begin{align}
D_{\Xt_i}&=\E[(\Xt_i-\E[\Xt_i|V])^2]\\
&=\E[\Xt_i^2]-\E[\E[\Xt_i|V]^2]\\
&=1-\E[\E[\Xt_i|V]^2],
\end{align}
and similarly
\begin{align}
D_{\Yt_i}&=1-\E[\E[\Yt_i|V]^2],
\end{align}
we have from Lemma 1 that (recall that $\E[\Xt_i\Yt_i]=\rho_i$)
\begin{align}
\rho_i^2&\le (1-D_{\Xt_i})(1-D_{\Yt_i})\\
\iff \rho_i^2 +D_{\Yt_i}(1-D_{\Xt_i})&\le (1-D_{\Xt_i})\\
\iff \frac{\rho_i^2}{1-D_{\Xt_i}} +D_{\Yt_i}&\le 1\\
\iff  D_{\Yt_i}&\le 1-\frac{\rho_i^2}{1-D_{\Xt_i}}.\label{cs}
\end{align}
Inserting \eqref{cs} into \eqref{conv1}, we find that minimizing \eqref{conv1} is equivalent to
\begin{align}
\mbox{maximize}&~~D_{\Xt_i}^{\la+1}\br{1-\frac{\rho_i^2}{1-D_{\Xt_i}}}=:f_\la(1-D_{\Xt_i}),
\end{align}
for $1\le i\le r$, where 
\begin{equation}
f_\la(x)=\frac{(1-x)^{\la+1}(x-\rho_i^2)}{x},\label{fdef1}
\end{equation}
and the maximization is carried out over functions $D_{\Xt_i}(\la)$, where $D_{\Xt_i}:\R_+\to[0,1]$.  Note that $D_{\Yt_i}(\la)\ge0$ and \eqref{cs} imply that $0\le D_{\Xt_i}(\la)\le 1-\rho_i^2\iff \rho_i^2\le (1-D_{\Xt_i}(\la))\le 1$.
In order to maximize the above expression, we set $f_\la'(1-D_{\Xt_i})=0$ to obtain
\begin{equation}
D_{\Xt_i}(\la)=1-\frac{\la\rho_i^2+\rho_i\sqrt{\la^2\rho_i^2+4(\la+1)}}{2(\la+1)}.
\end{equation}
Since $f(\rho_i^2)=f(1)=0$, $f(x)>0$ for $\rho_i^2<x<1$ and $f$ is smooth, this critical point must be the maximum.  Since the resulting value of $D_{\Xt_i}^{\la+1}D_{\Yt_i}$ was achieved by the Gaussian construction, we conclude that Gaussian auxiliaries suffice for achieving the optimal rate frontier.


\section*{Acknowledgment}

We would like to thank Jingbo Liu for insightful discussions. This work is supported by the National Science Foundation (grants CCF-1350595 and CCF-1116013) and the Air Force Office of Scientific Research (grant FA9550-12-1-0196).\vspace{-.1cm}


\bibliographystyle{ieeetr}

\bibliography{comminfo_isit}

\end{document}